\documentclass[nohyperref]{article}
\usepackage{microtype}
\usepackage{graphicx}
\usepackage{subfigure}
\usepackage{booktabs} 

\usepackage{hyperref}

\usepackage[accepted]{preprinticml2022}

\usepackage{amsmath}
\usepackage{amssymb}
\usepackage{mathtools}
\usepackage{amsthm}

\DeclareMathOperator{\E}{\mathbb{E}}

\newcommand{\indic}[1]{\mathbf{1}\left\{#1\right\}}

\newcommand{\probdist}{\mathbf{p}} 

\newcommand{\exposure}[2]{Q(#1,#2)}
\newcommand{\exposuresymbol}{Q}
\newcommand{\smoothexposuresymbol}{\mathcal{Q}}

\newcommand{\ZeroPercentComment}[1]{}
\newcommand{\potentialRemoval}[1]{}
\newcommand{\OldStory}[1]{}

\newcommand\coolrightbrace[2]{%
\left.\vphantom{\begin{matrix} #1 \end{matrix}}\right\}#2}
\newcommand\coolrightbracenvm[2]{ }
\usepackage[capitalize,noabbrev]{cleveref}

\usepackage{float}
\usepackage{caption}

\newtheoremstyle{break}
  {}
  {}
  {\itshape}
  {}
  {\bfseries}
  {.}
  {\newline}
  {}

\theoremstyle{break}
\newtheorem{theorem}{Theorem}
\newtheorem{proposition}[theorem]{Proposition}
\newtheorem{lemma}[theorem]{Lemma}

 \theoremstyle{definition}
  \newtheorem{definition}[theorem]{Definition}

\theoremstyle{remark}

\icmltitlerunning{Statistical anonymity}

\begin{document}

\twocolumn[
\icmltitle{Statistical anonymity:\\Quantifying reidentification risks without reidentifying users}

\begin{icmlauthorlist}
\icmlauthor{Gecia \mbox{Bravo-Hermsdorff}}{comp}
\icmlauthor{Robert \mbox{Busa-Fekete}}{comp}
\icmlauthor{Lee \mbox{M.~Gunderson}}{yyy}
\icmlauthor{Andr\'es Mun\~oz Medina}{comp}
\icmlauthor{Umar Syed}{comp}
\end{icmlauthorlist}

\icmlaffiliation{yyy}{Gatsby Unit, University College London, UK}
\icmlaffiliation{comp}{Google Research, New York, US}

\icmlcorrespondingauthor{\textit{Gecia \mbox{Bravo-Hermsdorff}}\\}{gecia@google.com}

\icmlkeywords{privacy, data anonymization}

\vskip 0.3in
]

\printAffiliationsAndNotice{}  

\begin{abstract}
Data anonymization is an approach to privacy-preserving data release 
aimed at preventing participants reidentification, 
and it is an important alternative to differential privacy 
in applications that cannot tolerate noisy data. 
Existing algorithms for enforcing $k$-anonymity in the released data 
assume that the curator performing the anonymization 
has complete access to the original data. 
Reasons for limiting this access range from 
undesirability to complete infeasibility. 
This paper explores ideas 
--- objectives, metrics, protocols, and extensions --- 
for reducing the trust that must be placed in the curator, 
while still maintaining a statistical notion of $k$-anonymity. 
We suggest 
trust (amount of information provided to the curator) and
privacy (anonymity of the participants)
as the primary objectives of such a framework. 
We describe a class of protocols aimed at achieving these goals, 
proposing new metrics of privacy in the process, 
and proving related bounds. 
We conclude by discussing a natural extension of this work 
that completely removes the need for a central curator. 
\end{abstract}

\section{Releasing private data (Background)}
\label{sec:intro}

As the use of big data continues to permeate modern society, 
so does the sharing of our personal data 
with centralized \mbox{third-parties}. 
For example, the U.S.~Census Bureau 
shares aggregated population statistics with lawmakers \cite{census}, 
and hospitals share medical information with insurance companies \cite{crellin2011survey}. 
If unregulated, 
this type of information poses a threat to individual privacy. 
A trivial way to completely protect the privacy of individuals 
would be to simply not share any of their information, 
but such an absolutist approach is neither feasible nor useful.  
A sensible compromise is to develop methods that 
balance the usefulness of the data against the privacy lost by the individuals.  

Two common frameworks for privacy-preserving data release are: 
\emph{differential privacy}, i.e., DP 
(and its various extensions, e.g., R\'enyi differential privacy) 
and 
\emph{$k$-anonymity} 
(and its various extensions, e.g., \mbox{$t$-closeness}). 

\subsection{A quick (incomplete) summary of DP} 
\label{sec:SummaryDP}

In the central model of differential privacy \cite{dwork2006calibrating}, 
a trusted curator stores the database, 
and an analyst\footnote{Note that here the ``analyst'' and the ``public'' are the same entity since the data observed by the analyst could be seen by anyone else.} 
issues queries about the database to a curator, 
who returns noisy responses.  
Such an approach requires the users to 
trust the curator with the entirety of their private data.  
Several models have been proposed to relax this requirement.

In the local model, 
each user adds noise to their own data and 
responds to the analyst directly \cite{evfimievski2003limiting}. 
In the shuffle model, 
each user encrypts their noisy data 
(such that only the analyst may read them), 
and sends them to a trusted shuffler.  
The shuffler then randomly permutes these encrypted messages 
before forwarding them to the analyst \cite{cheu2019distributed}.

\subsection{A quick (incomplete) summary of $k$-anonymity} 
\label{sec:SummaryAnom}

A dataset satisfies \mbox{$k$-anonymity} if for every individual whose data is contained in the dataset, 
their data are indistinguishable from that of at least \mbox{$k-1$} other individuals (also presented in this dataset). 
Since \mbox{$k$-anonymity} was first introduced \cite{sweeneykanon}, 
efficient algorithms for anonymizing a database (while preserving the maximum amount of information possible) have received increasing interest. 
Local suppression algorithms aim to achieve this by redacting specific (feature, user) entries of the database  \cite{meyerson2004complexity}, 
while global suppression algorithms redact the same set of features for every user \cite{el2009globally}. 

\cite{meyerson2004complexity} showed that the problem of optimally anonymizing a database by either local and global suppression is \mbox{NP-hard}. 
In light of these results, 
several approximation algorithms have been proposed, particularly for local suppression \cite{aggarwal2005approximation, GKOULALASDIVANIS20144}. 
Similar to the central model of differential privacy, 
these algorithms/curators require access to the entire private data.  
Unlike differential privacy, 
variants of \mbox{$k$-anonymity} that reduce the 
trust that participants must place in the curator remain relatively unexplored.

\section{Why we focus on \mbox{$k$-anonymity} (Motivation)}
\label{sec:whykanon}

Differential Privacy (DP) \cite{dwork2014algorithmic} is a measure of privacy loss (typically denoted by $\varepsilon$) that holds true \emph{no matter what} 
(e.g., even if additional information is released in the future).  
As a result of this strong propriety, any DP algorithm must be stochastic (e.g., by adding noise to the data).  
This, however, can be undesirable in a variety of applications (see Section~\ref{sec:ExamplesUse} for examples). 

In contrast, while \mbox{$k$-anonymity} can be satisfied without adding noise to the data,  
its the privacy guarantee are contingent on the availability of auxiliary information 
(see \cite{narayanan2008robust} for a famous example involving Netflix).

\subsection{There is no panacea for private data}
\label{sec:kworth}

As differential privacy offers an upper bound on each instance of privacy loss that holds regardless of anything else, 
it has a simple composition rule that can be invoked without further assumptions.  
Perhaps for this reason, DP is currently the \emph{de facto} academic definition of privacy.  

As the issues surrounding privacy become increasingly pressing societal issues, 
it seems natural that the entities managing our private data would like to offer meaningful privacy guarantees.  
Unfortunately, despite being the ``gold standard'', DP is often touted with essentially meaningless parameters \cite{domingo2021limits}.  
For example, the US census of 2020 claims a ``mathematical algorithm to ensure that the privacy of individuals is sufficiently protected'' with a ``budget'' of \mbox{$\varepsilon=19.61$} \cite{WinNT}. 
Setting aside a conspicuous similarity with the natural logarithm of the US population,\footnote{
The US population in 2020 is estimated at $331$ million, and $\ln\!\big((331\pm1)\cdot10^6\big)\approx19.617\pm0.002$.} 
the guarantee being made is essentially meaningless:  
``your participation in the census will not change the likelihood of any outcome by more than a factor of $331$ million.'' 

Given this clear rift in communication between theory and practice, 
it is fruitful to also consider privacy notions that might have fewer ``translation'' issues, despite their technically ``weaker'' guarantees.

\subsection{Natural extensions of $k$-anonymity}
\label{sec:extensions}

For simplicity, consider the following setting: 
A database is to be released containing i.i.d. samples from the population, 
and the values can be split into two disjoint sets ``Quasi-Identifiers'' (QI) and ``Sensitive Attributes'' (SA).  
QI are not known to an adversary \emph{a priori}, but could be learned (for some cost) via exogenous means.  
SA are features that are not known to the adversary, cannot be learned exogenously, and would be detrimental(valuable) to the participant(adversary) if learned by the adversary.  

Many ``$scalar$-word'' anonymity measures can be classified by the assumptions they make on the sensitive attributes (SA).  
The use of \mbox{$k$-anonymity} assumes that all SA are completely incomparable, while \mbox{$l$-diversity} \cite{machanavajjhala2007diversity} allows for the possibility of identical SA (but is still blind to the magnitude of differences).  
Other metrics, such as \mbox{$t$-closeness}, 
\mbox{$\delta$-disclosure}, and \mbox{$\beta$-likeness}, allow for a more general similarity metric between different SA \cite{khan2021k}.  

The main goal of this paper is to understand the \mbox{trade-off} between anonymization guarantees to the participants and the trust they must place in the entity performing the anonymization. 
We believe that \mbox{$k$-anonymity} is a suitable notion to use as a \mbox{proof-of-concept} to introduce such statistical relaxation. 
Extending this framework to more nuanced measures of anonymity would be of considerable practical interest.

\subsection{Application examples}
\label{sec:ExamplesUse}

Essentially, we consider a  setting in which the private variables (the Sensitive Attributes) are incomparable 
(i.e., there is no metric of similarity) and unique (no two private variables are identical).
In such a setting, $k$-anonymity is equivalent to $l$-diversity, and extensions such as $t$-closeness and $\beta$-likeness do not make sense (as the SA have no notion of similarity). 

For example,  
consider a database containing \mbox{X-rays} images (SA),   
along with some (Quasi-Identifying) demographics of the patients.  
The latter could likely be obtained by an adversary with minimal effort, 
whereas the former is essentially impossible to directly measure  
(without explicit cooperation from the individual).   
Given the exposing nature of these SA, 
it is not a stretch to think about an adversary using them for their personal gain  
at the expense of the owner of the images. 
Moreover, the details of everyone's insides are rather unique. 

Another application is that of preventing \textit{browser fingerprinting} \cite{laperdrix2020browser}. 
Malicious websites engaged in browser fingerprinting query detailed information about a user's device  
(e.g., which fonts they have installed).  
If these details are sufficiently unique, they can be used to covertly track a user across different the web.  
While certain system details can be made less amenable to fingerprinting by adding noise to them 
(e.g., window size/resolution), 
the option of returning noisy responses is often not practical (e.g., uncertain browser type).  
Several browsers have proposed to prevent fingerprinting by ensuring that the information queried by a website is always \mbox{$k$-anonymous}, and blocking the query otherwise.  
However, the only way to completely guarantee \mbox{$k$-anonymity} is to grant a central curator access to the full data of every user.

\section{The big picture (What we did)}
\label{Sec:TheBigPicture}

In seeking a version of \mbox{$k$-anonymity} that does not require a \mbox{fully-trusted} curator, it appears necessary to allow for some fraction of the database that does not satisfy \mbox{$k$-anonymity}.\\  
It is therefore tempting to make the following analogy with differential privacy: \\
\[
\arraycolsep=2pt
\begin{array}{lcl}
\textbf{Central $\varepsilon$-DP} & \longrightarrow & \textbf{Local ($\varepsilon,\delta$)-DP} \quad\text{\tiny(with shuffler)}\\
 & \big\Downarrow & \\
\textbf{Central $k$-Anon} & \longrightarrow & \textbf{Local ($k,Q$)-Anon} \quad\text{\tiny(with shuffler)} \\
 & & \text{\small(\textit{this paper})} \\
\end{array}
\]\\
where the ``amount of privacy'' is quantified by $\varepsilon$ and $k$, and the ``error rates'' by $\delta$ and $Q$.  
Here, we propose several ideas related to the bottom right.

\subsection{Main contribution}
\label{sec:contribution}

In this paper, we analyze ideas for anonymizing a database, 
while only granting the curator access to a partial view of the database.  
In such a setting, to publish any data, one must balance: 
\vspace{-6pt}
\begin{enumerate}
    \item \emph{Trust}: amount of information provided to the curator. \vspace{-2pt}
    \item \emph{Privacy}: anonymity of the participants. 
\end{enumerate}
\vspace{-6pt}
As \mbox{$k$-anonymity} can no longer be strictly guaranteed for all users, 
we quantify privacy using the \emph{exposure}: a new statistical version of \mbox{$k$-anonymity}, defined as the expected fraction of users who are not \mbox{$k$-anonymous} in the published database. 
We then upper bound the exposure with high probability.

\section{Problem setting and notation}
\label{Sec:ProblemSetting}

A private database is represented by a matrix $M$ consisting of $n$ rows and $d$ columns. 
Each \emph{row} in the matrix corresponds to a \emph{user}, 
and each \emph{column} corresponds to a \emph{feature} (e.g., age or gender). 

$M_{i\hspace{-0.08em}j}$ denotes the value of feature $j$ for user $i$. 
For any subset of features \mbox{$J \subseteq [d]$}, $M_{\hspace{-0.08em}J}$ denotes the submatrix containing only columns $J$. 
\mbox{$V_j = \{M_{i\hspace{-0.08em}j} : i \in [n]\}$} denotes the set of possible values for column $j$, 
\mbox{$V_J = \bigotimes_{j\in J} V_j$} the set of possible joint values for columns $J$, 
and \mbox{$V = \bigotimes_{j\in [d]} V_j$} the set of all possible joint values.

A user in a database $M$ is \emph{\mbox{$k$-anonymous}} if 
there are at least $k$ rows in $M$ that are identical to that user's row (e.g., \mbox{$k=1$} implies that their row is unique).  
If every user is \mbox{$k$-anonymous}, we say that $M$ is \mbox{$k$-anonymous}.  
\emph{Local} suppression algorithms achieve this by redacting specific \emph{entries} of $M$, 
while \emph{global} suppression algorithms redact entire \emph{columns}. 
In this paper, we consider global suppression, 
though we posit analogues of local suppression in the discussion (Section~\ref{sec:LocallyRedactingEntries}).


\section{Anonymization protocol}
\label{sec:Protocol}

We proposed a \mbox{two-step} protocol similar in spirit to the shuffle model of differential privacy \cite{cheu2019distributed}.  

In the first round, the users send messages to the curator.  
The curator learns only the marginal distribution of individual features (i.e., nothing about their correlations).  
Using this information, the curator selects an appropriate set of features to be released to the analyst. 

In the second round, 
the users send messages to the analyst. 
The analyst learns the full joint distribution between those features selected by the curator.  

\subsection{The first step}
\label{ssec:protocol_first_step}

First, the curator and shuffler both create public/private key pairs, 
sending the public keys to the users 
(see Figure~\ref{Fig:IllustrationProtocol}). 
Each user uses these public keys to encode one message for each feature.  
As identical values would result in identical encoded messages, 
the users first concatenate their value with a random string.\footnote{Thereby adding some cryptographic ``salt'' \cite{park2001cryptographic} to the receipt, if you will.}
The users then send these encoded messages to the shuffler.

The shuffler randomly permutes these encrypted messages, 
decodes their portion of the encryption,
and sends them to the curator.

The curator receives the messages,
decodes them, and remove the salt. 
The curator then selects a subset of features that are safe to give to the analyst.  
Such a decision can be made by using,
for example, composition rules (Theorems~\ref{Thm:CompositionTheoremSupport} and~\ref{Thm:CompositionTheoremC}) 
or statistical modeling (Section~\ref{sec:StatisticalDatabaseSetting}).

\subsection{The second step}
\label{ssec:protocol_second_step}

Using the same shuffling mechanism, 
the users now communicate with the analyst. 
Each user creates a \emph{single} message, 
encoding the entire subset of features that have been deemed ``safe'' to release by the curator. 

\begin{figure*}[!ht] 
\begin{center}
    \includegraphics[ clip, trim={-3cm -0cm -0cm -0cm},width=1\textwidth]{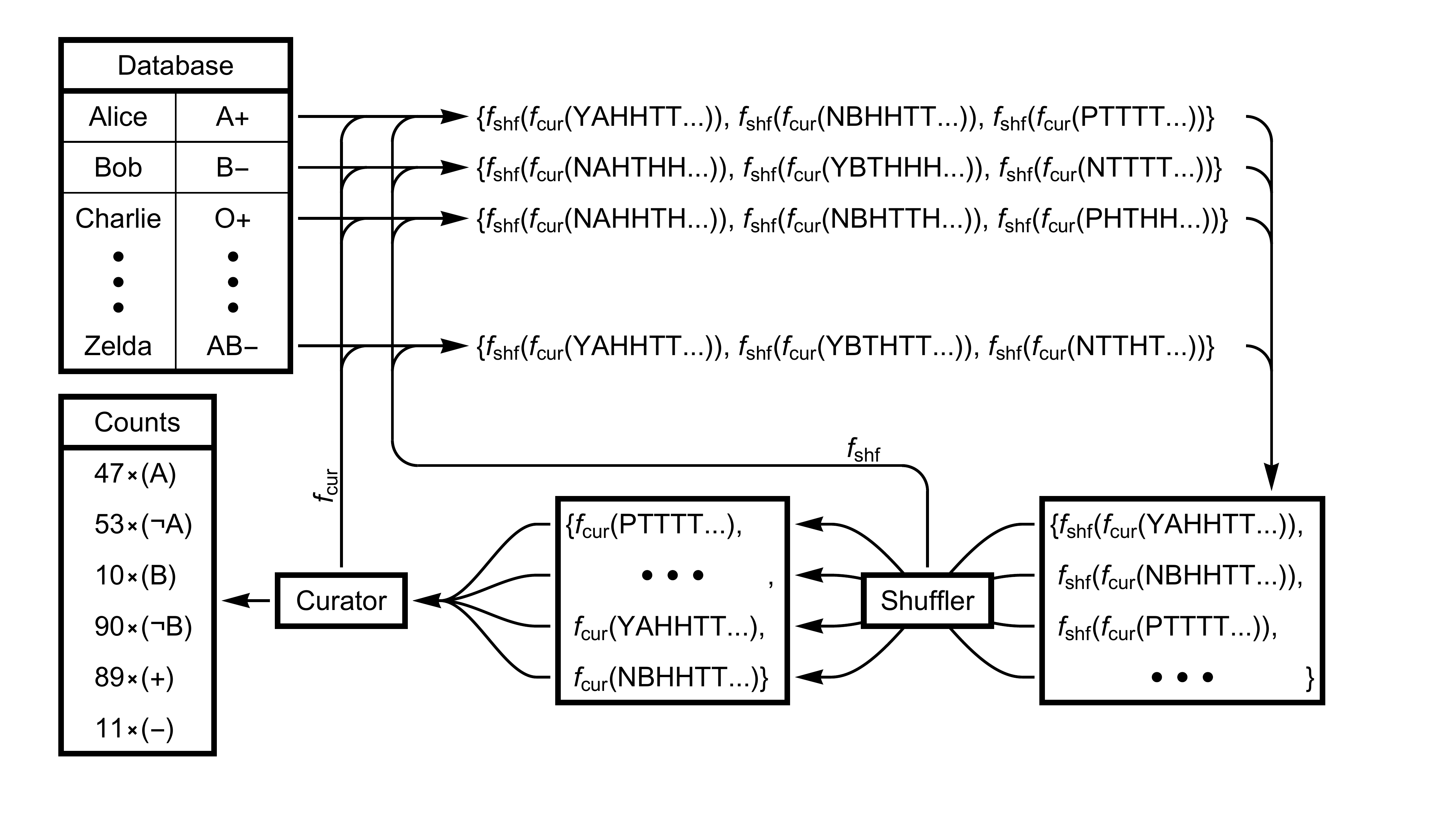}%
    \caption{
    \textbf{Anonymization protocol to obtain only the counts of individual features in a database.}\\
    Schema of how the curator can learn the counts of the individual features (i.e., the marginal distributions) without learning the full joint distribution (i.e., their correlations).  
    This protocol requires only the use of \mbox{public-key} cryptography as a primitive, and is example of a mix network \cite{sampigethaya2006survey}.\\
    For illustration, 
    we use a fictitious database of blood types.  
    A person's blood type (the joint variable) is determined by the value of 3 binary features (the marginal variables): 
    the presence or absence of antigens $A$, $B$, and $Rh(D)$ in their red blood cells. 
    As two examples, for Charlie's type of $O+$, the ``$O$'' corresponds to the absence of both the $A$ and $B$ antigens, 
    and the ``$+$'' to the presence of antigen $Rh(D)$, 
    whereas Zelda's type of $AB-$ corresponds to the presence of $A$ and $B$ and the absence of $Rh(D)$.  \\
    \textbf{\textbf{1.}} The protocol begins with the curator and the shuffler each generating their own public/private key pair, 
    and sending the public key ($f_{\textrm{\small{\vphantom{shf}cur}}}$ and $f_{\small{\vphantom{cur}\textrm{shf}}}$, respectively) to the participants for encrypting their data. \\
    \textbf{\textbf{2.}} Participants prepare each of their features by: 
    putting the value for the feature in a standardized form (e.g., YA, for presence of $A$, N for absence of $Rh(D)$, etc); adding a random suffix (e.g., HHTHT...);  
    then encrypting using first $f_{\textrm{\small{shf}}}$, 
    and then $f_{\textrm{\small{\vphantom{shf}cur}}}$. 
    The participant then send these messages (one for each feature) to the shuffler. \\
    \textbf{\textbf{3.}} The shuffler secretly shuffles all the messages, 
    then decrypts each of them using their own key, 
    and passes these (now singly encrypted) messages to the curator. \\
    \textbf{\textbf{4.}} The curator decrypts the message, 
    obtaining the counts for each feature. 
    }
    \label{Fig:IllustrationProtocol}
\end{center}
\vskip -0.2in
\end{figure*}


Within this setting, 
we analyze two situations: 
\vspace{-6pt}
\begin{enumerate}
    \item \emph{Fixed database} (Section~\ref{sec:FixedDatabaseSetting}): the \emph{same users} participate in the first and second step of the protocol. 
    \item \emph{Statistical database}  (Section~\ref{sec:StatisticalDatabaseSetting}): \emph{different users} participate in the first and second step of the protocol. 
\end{enumerate}

\section{Fixed database setting}
\label{sec:FixedDatabaseSetting}

We start by analyzing the simpler setting, 
where there is a fixed database, 
with the same set of users sending their individual entries to the curator
and their redacted rows to the analyst.  
Therefore, the curator decides which feature the users should redact when sending their data to the analyst based on complete knowledge of the marginal distribution of the features. 
The \mbox{X-ray} example we mentioned in Section~\ref{sec:ExamplesUse} is an application for this setting. 

\subsection{Exposure (new privacy measure)}
\label{sec:exposure}


\begin{definition}[\textbf{Exposure}]
\label{def:exposure}
For a given probability threshold \mbox{$t \in [0, 1]$}, 
the \emph{exposure $\exposuresymbol(t)$} is defined as the fraction of users that are less than \mbox{$(tn)$-anonymous} in matrix/database $M$: 
\begin{equation}
\exposuresymbol(t) = \sum_{v \in V} \probdist(v) \cdot \indic{\probdist(v) < t},
\end{equation}
where
\begin{equation}
\probdist(v) = \frac{|\{i \in [n] : M_{ij} = v_j \textrm{ for all } j \in [d]\}|}{n}
\end{equation}
is the empirical (observed in $M$) joint distribution of possible outputs $V$.  
\end{definition}

For subset of columns $J$ or a single column $j$, 
we define 
\[
\exposuresymbol_J(t), \probdist_J(v), \exposuresymbol_j(t), \probdist_j(v),
\]
as the restriction of the above definitions to those column(s).

When the database to be released to the analyst is different than the one observed by the curator (Section~\ref{sec:StatisticalDatabaseSetting}), 
we put a hat on quantities belonging to the latter, e.g., $\widehat{\exposuresymbol}_J$ and $\widehat{\probdist}_J$. 

A central question of our paper is: 
\vspace{-6pt}
\begin{quote}
    \emph{Given a subset of columns \mbox{$J \subseteq [d]$},\\
    empirical distribution $\probdist_j$ for each $j \in J$,\\
    and probability threshold $t \in [0, 1]$,\\
    estimate the value of $\exposuresymbol_J$(t).}
\end{quote}
In the next sections we answer this question in two parts:
\vspace{-6pt}
\begin{itemize}
    \item In Section~\ref{sec:composition}, we upper and lower bound\\
    $Q_J$, the exposure of a set of features, in terms of the\\ 
    $Q_j$'s, the exposure of the individual features \mbox{$j \in J$}. 
    \item In Section~\ref{sec:estimation}, we upper and lower bound\\
     $Q_j$, the exposure of an individual feature, in terms of\\ 
     $\widehat{\exposuresymbol}_j$, its observed exposure in another sample.
\end{itemize}

\subsection{Composition theorems (from marginals to joint)}
\label{sec:composition}

As a simple adversarial example to illustrate the hardness of relating the exposure of individual columns to the exposure of their joint, 
consider the following $(n+1)$-by-$2$ binary database $M$:
\begin{equation*}
M = 
\begin{bmatrix}
1 & 0 \\
1 & 0 \\
\vdots &\vdots \\
0 & 0 \\
0 & 1 \\
0 & 1 \\
\vdots &\vdots \\
\end{bmatrix}%
\begin{matrix}%
    \coolrightbrace{1 \\ 1\\ \vdots}{\frac{n}{2}}\\
    \coolrightbracenvm{0 }{ }\\
    \coolrightbrace{0 \\ 0\\ \vdots }{\frac{n}{2}}
\end{matrix}
\end{equation*}
    It is easy to see that this matrix is \mbox{$\frac{n}{2}$-anonymous} for column $1$ and column $2$ individually. 
However, 
notice that the ``middle'' user is unique, so the entire matrix $M$ is only \mbox{$1$-anonymous}.

This example shows that even when the anonymity of each column is high, 
there is no general strict guarantee that can be provided for the anonymity of all users. 
It also suggests that, 
when the anonymity of the individual columns is high, 
the fraction of users for which the anonymity is violated is small. 
The next theorems precisely quantify how large this fraction can be.  

First, we upper bound the exposure of the empirical joint distribution of features $\probdist_J$ in terms of 
the exposure and support size of each individual feature distribution $\probdist_j$. 
\vspace{4pt} 

\begin{theorem}[\textbf{Composition with known support sizes}] 
For any subset of columns \mbox{$J \subseteq [d]$} , 
probability thresholds \mbox{$\{t_j\}_{j \in J}$}, 
and any \mbox{$j^{\star} \in J$}:
\begin{align*}
    \exposuresymbol_J\left(\prod_{j \in J} t_j\right) &\leq \sum_{j  \in J} \exposuresymbol_j(t_j) + \sum_{j\in J \setminus \{j^\star_{ }\} } t_j|V_j| 
\end{align*}
\label{Thm:CompositionTheoremSupport}
\end{theorem}
This bound is particularly useful when the support size of each column is small. 
When their support size are large or unknown, 
its usefulness deteriorates. 
The next theorem replaces the support size with a free parameter $c$ that can be optimized by the curator to decide which columns should be redacted when doing the global suppression.
\vspace{4pt} 

\begin{theorem}[\textbf{General composition rule}] 
For any subset of columns $J \subseteq [d]$, 
probability thresholds $\{t_j\}_{j \in J}$, 
and free parameter $c \in (0, 1)$:
\begin{align*}
    \exposuresymbol_J\!\left(c\prod_{j \in J} t_j\right) \leq  &\sum_{j \in J} \exposuresymbol_j(t_j) +  c
\end{align*}
\label{Thm:CompositionTheoremC}
\end{theorem}
These lower bounds on the probability threshold $t$ for the joint distribution given by Theorems \ref{Thm:CompositionTheoremSupport} and \ref{Thm:CompositionTheoremC} depend on the product of the probability thresholds for the marginal distributions. 
Clearly, 
this product decays exponentially with the number of marginal distributions $|J|$, 
which means the guarantees from these theorems become weaker when $|J|$ is large. 
Also notice that these guarantees require an additional \emph{slack}, 
either in terms of the support size or the tunable parameter $c$. 
The following theorem show that this slack term is necessary in general.

\begin{theorem}[\textbf{The ``slack'' term is necessary}]
Let \mbox{$c \in (0, 1]$} such that $\frac{1}{c}$ is an integer. 
There exists a matrix $M$, subset of columns \mbox{$J \subseteq [d]$}, and probability thresholds probability thresholds \mbox{$\{t_j\}_{j \in J}$}, such that:
\begin{align*}
    Q_J\!\left(c\prod_{j \in J} t_j\right) &\ge \sum_{j \in J} Q_j(t_j)\\
\end{align*}
\label{Thm:SlackNecessaryOnComposition}
\end{theorem}

\subsection{Applications to \mbox{real-world} data}
\label{sec:application}

As a concrete example, 
we applied the above bounds
(Theorems~\ref{Thm:CompositionTheoremSupport} and~\ref{Thm:CompositionTheoremC}) 
to a dataset from the UCI repository \cite{murphy1992uci}, 
containing data extracted from the 1994 US Census \cite{kohavi1996scaling}. 
It contains $32561$ users and $14$ features (in our analysis we used $4$ of those features,
with support sizes ranging from $2$ to $9$). 
Figure~\ref{fig:adultdataset} displays both the true exposure 
and our bound for the exposure obtained by using both Theorems~\ref{Thm:CompositionTheoremSupport} and \ref{Thm:CompositionTheoremC}, 
and taking the minimum.

\begin{figure*}[!ht]
\begin{center}
\includegraphics[trim={0cm 0.1cm 0cm 0cm},clip, scale=0.43]{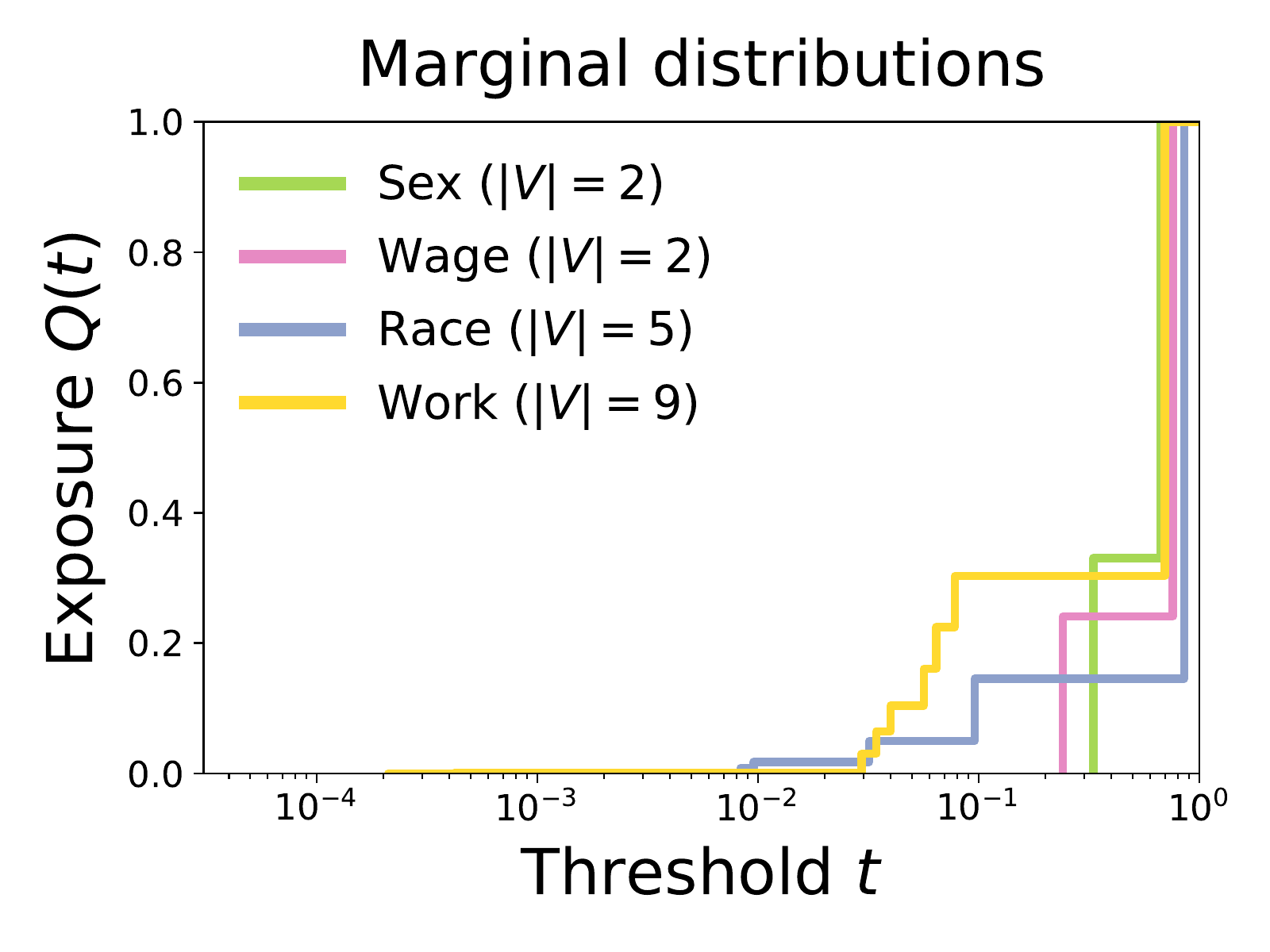}
\hspace{10pt}
\includegraphics[trim={0cm 0.1cm 0cm 0cm},clip, scale=0.43]{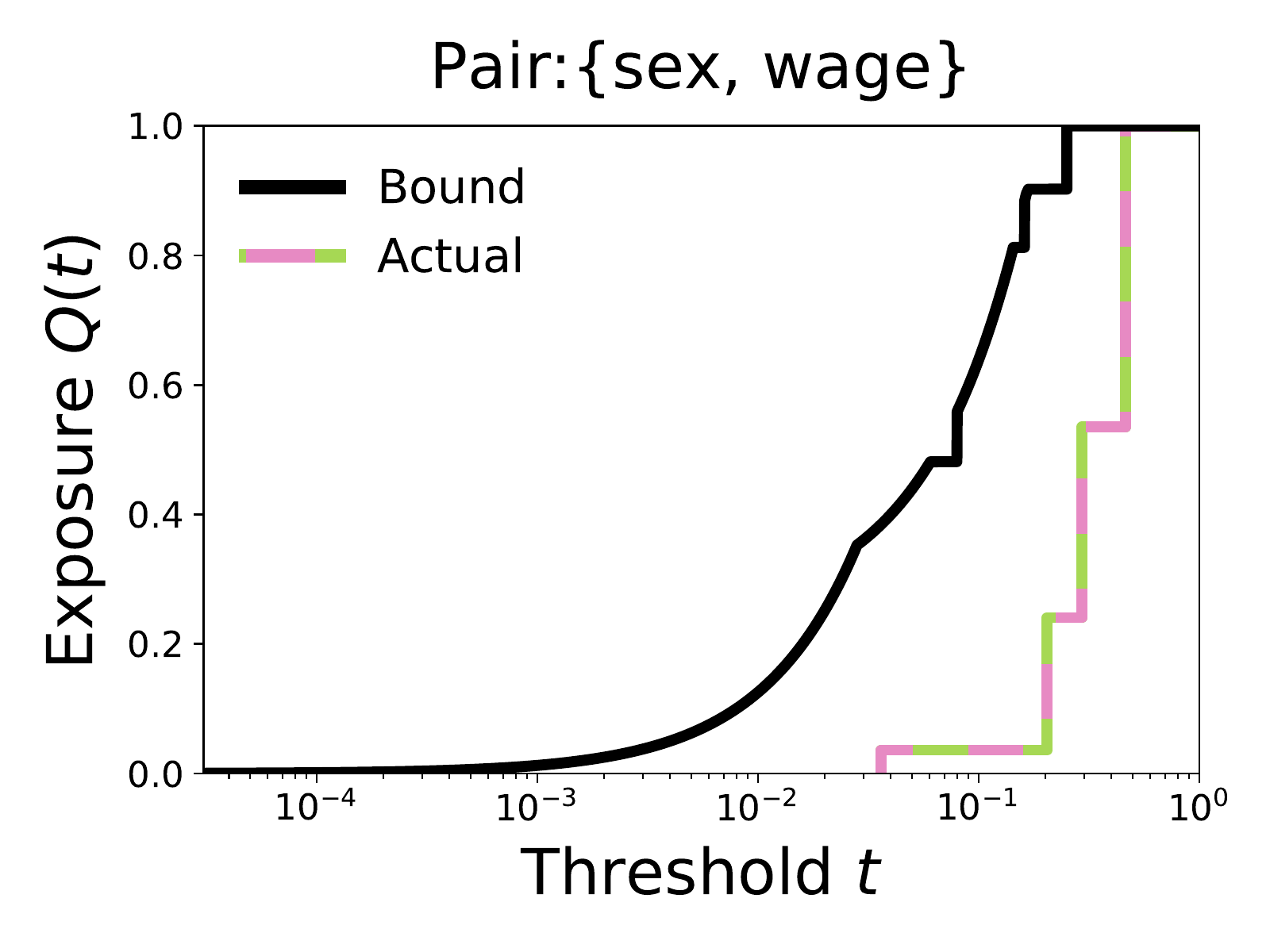}\\
\includegraphics[trim={0cm 0.3cm 0cm 0cm},clip, scale=0.43]{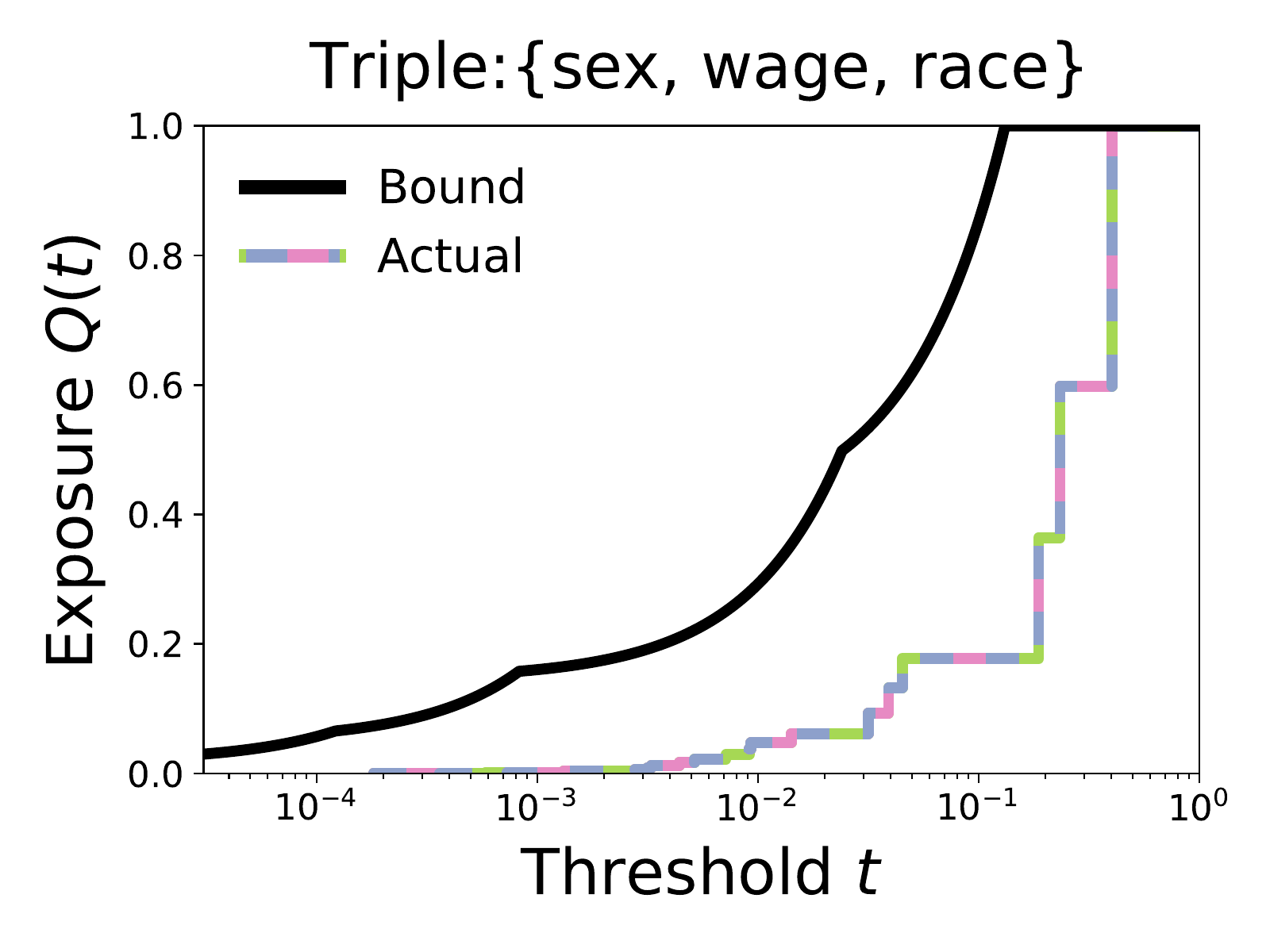}
\hspace{10pt}
\includegraphics[trim={0cm 0.3cm 0cm 0cm},clip, scale=0.43]{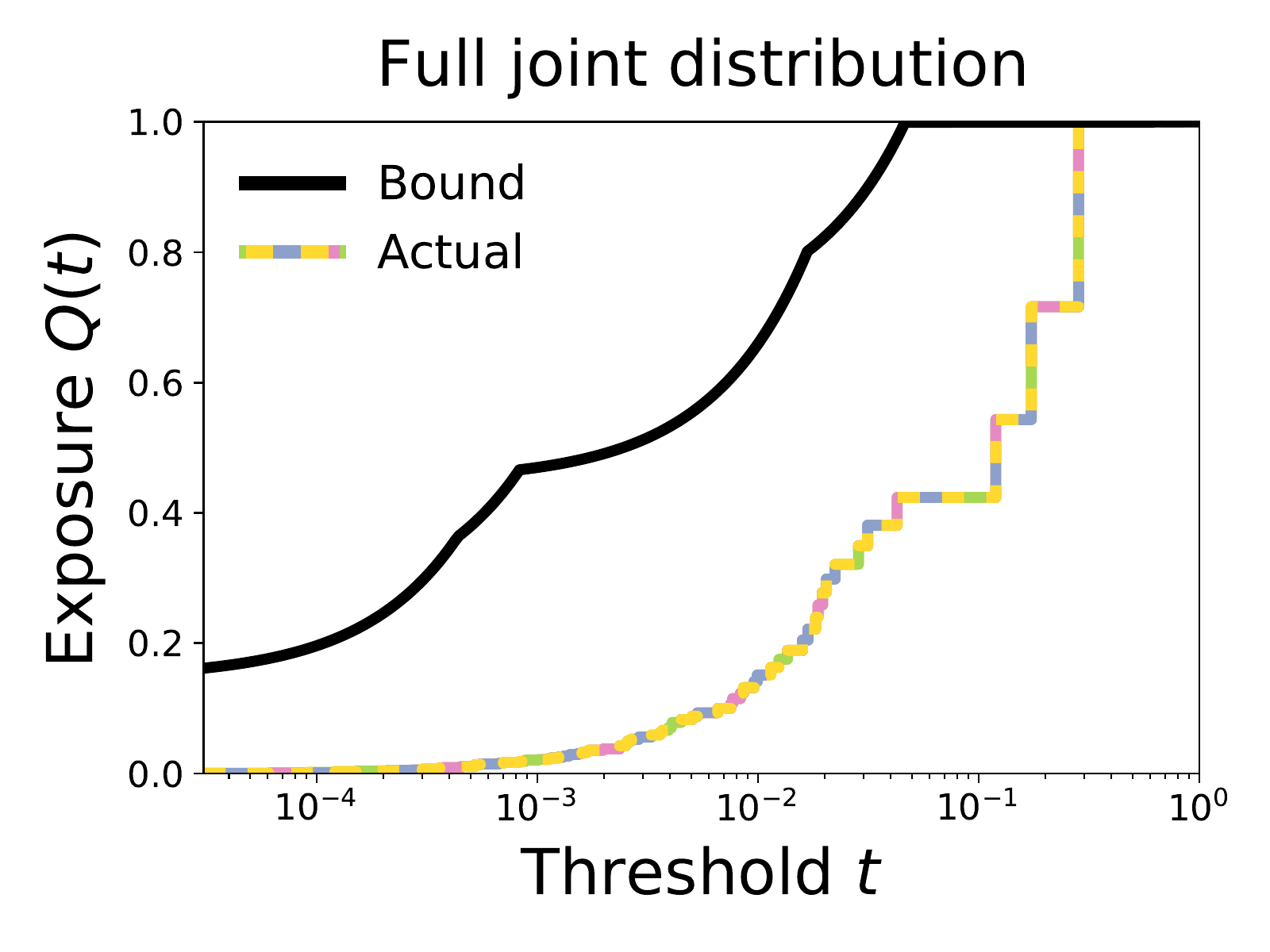}
 \caption{\textbf{Examples of exposure curves for real-world data.}\\ 
 Here, we consider $4$ features from the 1994 US Census dataset \cite{kohavi1996scaling}. 
 In order of appearance, the number of different output values for that feature are: $2$ (sex), $2$ (wage, thresholded at $50k$ per year), $5$ (race), and $9$ (work, e.g., Federal, self-employed, never worked). 
  Upper left shows the exposure curves for the frequencies of these individual features. 
 The remaining plots show how the exposure curve (lighter/colored lines) changes as we consider the joint distribution of an increasing number of features.  
 In black is the minimum of all possible upper bounds given by Theorems~\ref{Thm:CompositionTheoremSupport} and \ref{Thm:CompositionTheoremC}.  
}
 \label{fig:adultdataset}
 \end{center}
\vskip -0.2in
\end{figure*}

\section{Statistical database setting}
\label{sec:StatisticalDatabaseSetting}

We now turn attention to the setting where the values of the individual features and the data to be released do not come from the same set of users. 
To model such a setting, 
we assume that the data from both sets of users are sampled from the same underlying distribution.  
The browser fingerprinting example we mentioned in Section~\ref{sec:ExamplesUse} is an application for this setting. 

\subsection{Estimation of the exposure}
\label{sec:estimation}

In this case, the curator can estimate the distribution of values across users. 
This distribution can then be used to estimate the exposure for the set of users who will release the data to the analyst. 
It is natural to ask what guarantees can be given on the exposure of this data. 
Our first result is a bound on the \mbox{plug-in} estimator of the exposure $\widehat{\exposuresymbol}$. 

\begin{theorem}[\textbf{Plug-in estimator for the exposure}] 
\label{thm:plugin-exposure}
Let \mbox{$\gamma > 0$} and \mbox{$\delta \in [0, 1]$}.\\
If the number of samples is \mbox{$n \ge \frac{\log\frac1\delta + \log |V_j| }{2\gamma^2}$},\\ 
then with probability at least \mbox{$1 - \delta$}:\\
\[
\widehat{\exposuresymbol}_j(t - \gamma) - \gamma|V_j| \le \exposuresymbol_j(t) \le \widehat{\exposuresymbol}_j(t + \gamma) + \gamma|V_j| .
\]

\end{theorem}

This theorem provides us with a way to quantify the difference between the empirical estimator of exposure $\widehat{\exposuresymbol}_j$ and the true exposure $\exposuresymbol_j$.  
Notice however that due to the discontinuity of $\widehat{\exposuresymbol}_j$ the bound can become vacuous even for small values of $\gamma$. 
The following lower bound shows that this issue is true of any estimator of the exposure, 
not only the \mbox{plug-in} estimator. 
\vspace{4pt}

\begin{theorem}[\textbf{Hardness of estimating the exposure}] \label{thm:lower}
Let\mbox{ $\Delta_s = \{ \sum_{i=1}^s p_i = 1 : p_i \ge 0 ~\forall i \}$} be the set of distributions over \mbox{$\{1, \dots, s \}$} 
and \mbox{$\mathcal{F}_n$} be the the set of measurable functions mapping \mbox{$\{ 1, \dots , s \}^n \mapsto [0,1]$}. 
Then
\begin{align*}
    \lim_{n\to \infty} \mathcal{R}_n 
    \geq \frac{1}{16s}
\end{align*}
where the minimax risk $\mathcal{R}_n$ of exposure is defined as
\begin{align*}
\inf_{f \in \mathcal{F}_n} \sup_{\probdist \in \Delta_s }
    \E_{x_1, ..., x_n \sim \probdist}\Bigg[\sup_{t \in [0,1]}\Big|f(x_1, ... , x_n ) - Q_{\probdist}( t ) \Big|\Bigg] \enspace .
\end{align*}    
\end{theorem}

This formalizes the fact that 
the exposure function cannot be estimated with arbitrary small error.  
More concretely, think of the case when the threshold is equal to one of the probabilities $p_i$. 
In this case, 
one would need to estimate the probability of that observation with zero error to get an estimate of the true exposure with also zero error, 
which clearly requires infinitely many samples. 
This issue is inherent to several measures that are based on thresholded statistic of a cumulative distribution function (CDF) such as quantiles \cite{chen2020survey}.  
The next section offers a possible remedy.

\subsection{Statistical exposure (new privacy metric)} 
\label{sec:StatExposure}

As we have just shown, estimating the exposure of a random database from samples is hard. 
We leverage the assumption that the observed and released databases are sampled independently from the same distribution $\probdist_{ }$, and instead estimate the \emph{statistical exposure} (see also Appendix~\ref{Appendix:statistical_exposure}): 

\begin{definition}[\textbf{Statistical exposure}]
The \emph{statistical exposure}, \mbox{$\smoothexposuresymbol_{\probdist}(n,k)$}, is the probability that a random user in a database of size of $n$ sampled i.i.d. from the discrete probability distribution $\probdist_{ }$ is less than \mbox{$k$-anonymous}:\footnote{
The CDF of the binomial distribution, $\probdist(x\leq k)$,  
is given by $I_{1-p}\big(n-k,k+1\big)$.
For a given user in a database of size $n$ to be less than \mbox{$k$-anonymous},
there must be at most \mbox{$k-2$} with the same features out of \mbox{$n-1$} other i.i.d. samples, 
hence the term $I_{1-p}\big( (n-1)-(k-2), (k-2)+1\big) = I_{1-p}\big( n-(k-1), k-1\big)$.
}
\begin{align}
     \smoothexposuresymbol_{\probdist}(n,k) := \sum_{i=1}^{|V|}  p_i I_{1-p_i}\big( n-(k-1), k-1 \big),
\label{eq:smoothexposure}
\end{align}
where the function $I$ is the regularized incomplete beta function:
\begin{align*}
     I_{p}(a,b) &\equiv \frac{B(a,b;p)}{B(a,b;1)}, \\
     B(a,b;p) &\equiv \int_{0}^{p}z^{a-1}(1-z)^{b-1}dz. 
\end{align*}
\end{definition}

The main advantage of the statistical exposure is that, unlike exposure, 
it can be accurately estimated with access to an estimate of $\probdist$ as shown by the following theorem.
\vspace{4pt}

\begin{theorem}[\textbf{Estimating the statistical exposure}]
\label{thm:approximation_smooth}
Let $\probdist$ be the true distribution, 
and $\widehat{\probdist}$ its empirical frequency. 
Then, for all $k$ and $n$: 
\begin{equation}
    |\smoothexposuresymbol_{\probdist}(n,k) - \smoothexposuresymbol_{\widehat{\probdist}}(n,k)| \leq C\sqrt{n}\|\probdist - \widehat{\probdist}\|_\infty,
\end{equation}
where $C$ is a constant that depends linearly on the support size of $\probdist$ (see Appendix~\ref{prooftheorem6}).
\end{theorem}

Figure~\ref{Fig:SmoothExposureSingle} illustrates the difficulty of estimating the exposure from samples, 
and that the statistical exposure is a more reliable estimator with less variance. 
 \begin{figure}[h]
 \begin{center}
  \includegraphics[trim={0 0cm 0 0cm},clip, scale=0.43]{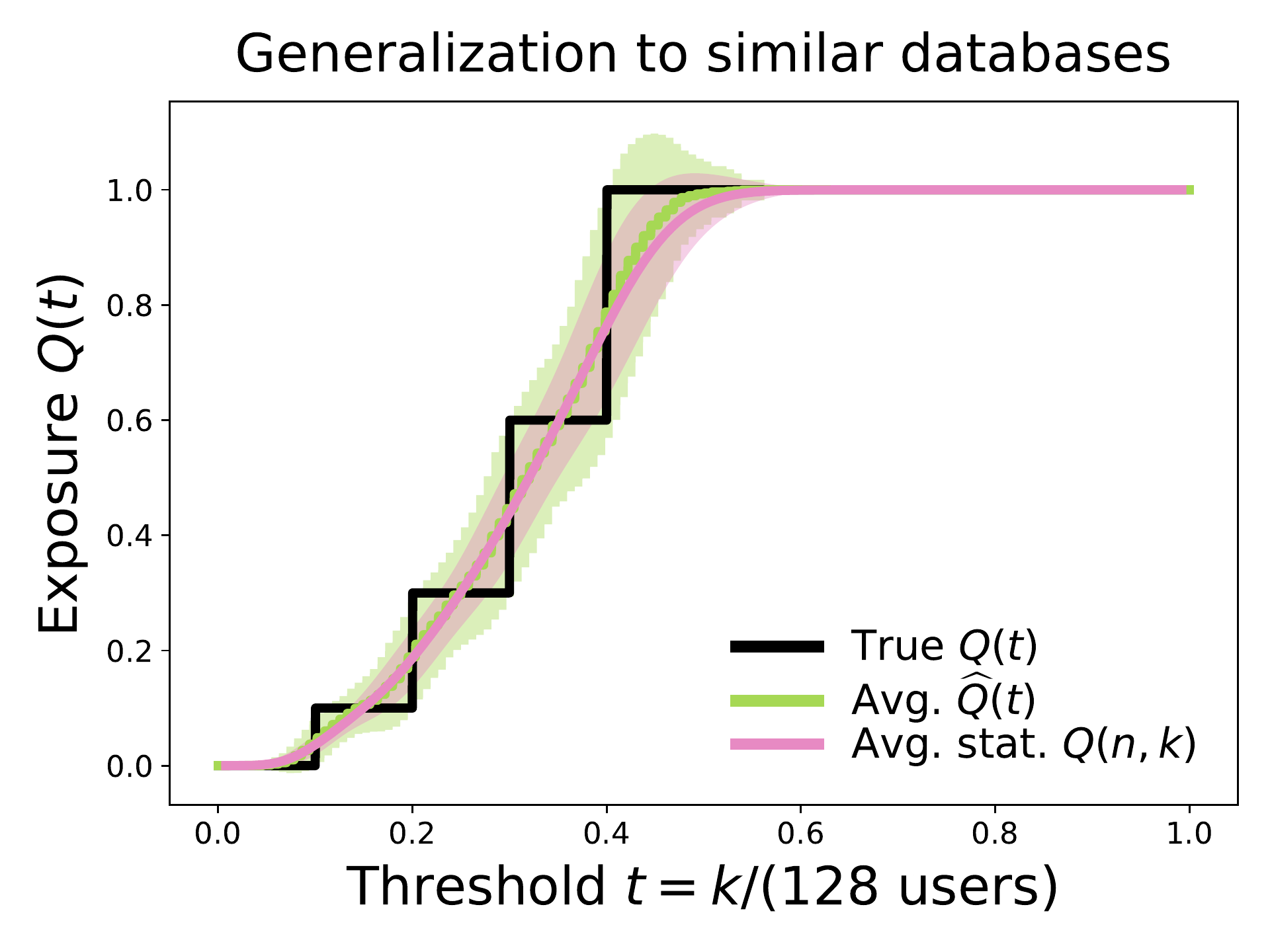}
  \caption{
  \textbf{The statistical exposure is a more robust\\estimator of the true (normalized) anonymity.}\\
  We ran $1000$ simulations of a database with $128$ users sampled i.i.d.~from a multinomial distribution $\probdist$ over $4$ outputs, and computed the \mbox{plug-in} estimators for the exposure (green) and the statistical exposure (pink). 
  When using one database to estimate the exposure of another one drawn from the same distribution, one expects fluctuations. 
  By taking these into account, the statistical exposure has a smaller standard deviation (shading) than the \mbox{plug-in} estimator for the exposure.  
  }
  \label{Fig:SmoothExposureSingle}
   \end{center}
   \vskip -0.2in
\end{figure}

\subsection{Relationship with Shannon entropy\\
\hphantom{7.3.}(a metric used for browser fingerprinting)}
\label{sec:entropyconnection}

The browser fingerprinting community has frequently used the Shannon entropy as a metric of how identifiable users are given the
information provided by different Application Programming Interfaces (APIs) (e.g., screen size, browser type, font installed). 
Given a distribution over values in a database \mbox{$\probdist = (p_1, \ldots, p_m)$}, 
its Shannon entropy  $H(\probdist) $ is defined as: 
\begin{equation*}
    H(\probdist) = -\sum_{i=1}^m p_i \log(p_i) = \mathbb{E}_{I\sim \probdist}[\log(p_i)].
\end{equation*}
One appealing property of the entropy is that marginal entropies can be used to bound the entropy of the full distribution via \mbox{inclusion-exclusion} principles.  
In addition, it can be accurately estimated with even when the number of samples is less than the support size of the distribution \cite{minmaxentropy}. 

However, a common (erroneous) interpretation of entropy is that if a database consists of $n$ users and the entropy is $B$ bits, 
then each value in the database is shared by \mbox{$\sim n/2^B$} users. 
While encouraging, this interpretation is true only when $\probdist$ is uniform (or close to it). 
The following proposition shows a more accurate interpretation of the entropy by using the exposure. 

\begin{proposition}[\textbf{Exposing the entropy}]
\label{prop:entropyconnection}
Let $\probdist$ denote a distribution and $Q(t)$ denote its exposure at threshold $t$. 
The following relations between entropy and $Q(t)$ hold:
\begin{align*}
    H(\probdist) = \int_0^1 \frac{Q(t)}{t} dt \qquad \text{and} \qquad
    Q(t) \leq -\frac{H(\probdist)}{\log (t)} .
\end{align*}
\end{proposition}
Concretely, consider the scenario of a database with \mbox{$n = 2^{16}\sim 65000$} users and a distribution \mbox{$\probdist$ with $H(\probdist) = 8$}. 
The common interpretation would suggest that the majority of users cannot be identified up to \mbox{$2^8 = 256$} users. 
On the other hand if \mbox{$t = 2^{-16}$}, 
then $Q(t)$ corresponds to the fraction of users who can be uniquely identified,
and the above bound implies that this could be up to $50\%$ of the users. 
In the Appendix Section~\ref{appendix:connectionentropy}, we show this bound is tight.

\section{Extensions}
\label{sec:FutureDirections}

\subsection{Locally redacting entries}
\label{sec:LocallyRedactingEntries}

In this paper, 
we considered global suppression; 
the curator redacts the same set of features for all users.  
Local suppression methods allow for a more targeted preservation of privacy, 
allowing for greater utility of the released database.  
To implement this within the presented framework, 
the curator would issue a set of conditional statements 
(e.g., ``If you have value $v^\star$ for feature $1$, redact it'').  

One must be careful however, 
as the act of redacting a value itself contains some information.  
Indeed, 
in the example above, 
if only one user has value $v^\star$ for feature $1$, 
nothing has changed.  
This could be overcome by asking the other users to randomly redact their value for that feature as well.  

\subsection{A hierarchical protocol}
\label{sec:HierachicalProtocol}

Throughout this paper, 
we have considered a \mbox{two-step} process: 
the curator learns about the marginal distributions, 
then informs the users what to redact when giving their information to the analyst.  
The set of marginal distributions is a rather \mbox{coarse-grained} view of the full joint distribution. 
To obtain a more precise picture, the curator requires some knowledge of the correlations between different feature values.  
This could be accomplished by allowing the curator to ask a sequence of queries with increasing complexity.  
At each step, the curator uses their current understanding of the distribution to decide which queries are likely ``safe'' to ask.  

Concretely, 
consider quantifying correlations of increasing order.  
In the first round, 
the curator asks for individual feature values,
thereby learning the expected frequency of each (the mean).  
In the second round, 
the curator asks the users to provide specific \emph{pairs} of feature values,
thereby learning something about the \emph{covariance} between them.  
From the first round, 
the curator knows that some feature values are shared by only a few users.  
Asking about pairs involving those values is likely to cause significant privacy loss, 
so the curator specifically does not ask them of the users. 
Using covariance information from the second round, 
the curator asks for specific \emph{triples} of feature values.  
This process could then continue until there are no ``safe'' queries.  
Granted with this more \mbox{fine-grained} picture of the distribution, 
the curator would make a decision as to how the users should release their data to the analyst.

\subsection{Distributed private data aggregation}
\label{sec:distributedprivatedataagg}

By taking a statistical approach to reducing trust requirements, 
the doors open for many exciting applications. 
For instance, 
imagine a group of users, 
all with their own private data, 
would like to know something about their collective statistics without compromising the privacy of any individual.
The protocol in the previous sections could be used nearly verbatim, 
only now the participants, the curator, and the analyst, are all the same entity. 
With zero trust invested in anything but the protocol, such a method of distributed private data aggregation could prove to be a very useful tool.

\section{Coda}
\label{sec:coda}

To close, we remark that, 
as with any new idea, 
frameworks that claim to guarantee some level of privacy should be treated with caution.  
Does the metric capture the notion you are trying to quantify?  
It is difficult to judge ``how private'' something is if it is measured incorrectly. 
Are there proven bounds for this metric? 
Guarantees aren't worth much if they are frequently false. 
To combat misrepresentation, 
either malicious or accidental, 
it is imperative that the problems being solved are appropriately practical, 
and that metrics used to evaluate performance are appropriately statistical.


\bibliography{example_paper}
\bibliographystyle{icml2022}

\newpage
\appendix
\onecolumn

\section{Composition rule with known support size (Theorem~\ref{Thm:CompositionTheoremSupport})}

For this proof, we use the following two lemmas:

\begin{lemma}[\textbf{Logsum inequality}]
\label{lemma:logsum}
Let\mbox{$f\colon \mathbb{R}_+ \to \mathbb{R}$} be a function such that \mbox{$g(p) = p f(p)$} is a concave function. \\
If \mbox{\mbox{$((a_i,b_i))_{i=1}^n > 0$}}, then:
\begin{equation*}
    \sum_{i=1}^n a_i f\left(\frac{a_i}{b_i}\right) \leq
    \left( \sum_{i=1}^n a_i \right) f\left(
    \frac{\sum_{i=1}^n a_i}{\sum_{i=1}^n b_i}\right)
\end{equation*}

\begin{proof}
Let \mbox{$B= \sum_{i=1}^n b_i$}, then we have:
\begin{equation*}\sum_{i=1}^n a_i f\left(\frac{a_i}{b_i}\right)
= B \sum_{i=1}^n \frac{b_i}{B} \frac{a_i}{b_i} f\left(\frac{a_i}{b_i}\right) = B \sum_{i=1}^n \frac{b_i}{B} g\left(\frac{a_i}{b_i}\right)
\end{equation*}
Since $g$ is a concave function, 
we can use Jensen's inequality to bound the previous expression by: 
\begin{equation*}
\leq B g\left(\sum_{i=1}^n \frac{b_i}{B} \frac{a_i}{b_i}\right) = B g\left( \frac{1}{B}\sum_{i=1}^n a_i\right) = 
\left(\sum_{i=1}^n a_i\right)  f\left(
    \frac{\sum_{i=1}^n a_i}{\sum_{i=1}^n b_i}\right)
\end{equation*}
\end{proof}
\end{lemma}

\begin{lemma}
\label{lemma:hq}
Let $G_q$ be a function \mbox{$G_q \colon [0,1] \to [0,1]$} defined as:
\begin{equation*}
G_q(p) = \min \bigg(1, \frac{q(1 - p)}{p(1-q)} \bigg)
\end{equation*}
Then, for every \mbox{$q \in [0,1]$}: 
\begin{itemize}
    \item The function $p \mapsto p G_q(p)$ is concave 
    \item  $p G_q(p) \geq p$ (for $p<q$), and \\
    $p G_q(p)<q$ (always).
\end{itemize}
\end{lemma}
\begin{proof}
It is easy to see that \mbox{$p G_q(p) = \min \big(p, \frac{q(1-p)}{1-q} \big)$} is concave, as it is the minimum of two linear functions.\\ 
These two linear functions are equal at $q$, so the maximum of $p G_q(p)$ is $q$, with argument $p=q$. 
\end{proof}

Now, we first prove the theorem for two marginal distributions: creatively named ``$1$'' and ``$2$'':
\begin{equation}
\label{eq:comp_two}
\exposuresymbol_J(q_1 q_2) \leq \exposuresymbol_1(q_1) + \exposuresymbol_2(q_2) + q_2 |V_2|,
\end{equation}
where $|V_2|$ is the size of the support of marginal $2$. 
\begin{proof}
By definition of exposure we have:
\begin{align}
1 - \exposuresymbol_J(q_1 q_2) 
&= \sum_{v_1}\sum_{v_2} p(v_1, v_2) \indic{p(v_1, v_2) \geq q_1 q_2}\nonumber\\
&\geq \sum_{v_1}\sum_{v_2} p(v_1, v_2) \indic{p(v_1, v_2) \geq q_1 q_2}\indic{p(v_1) \geq q_1} 
\label{eq:initial_eq}
\end{align}
Notice that if $p(v_1, v_2)/p(v_1) \geq q_2$ and $p(v_1) \geq q_1$ then $p(v_1, v_2) \geq q_1 q_2$.\\ 
Thus, we can lower bound the above expression by 
\begin{flalign*}
    &\sum_{v_1}\sum_{v_2} p(v_1, v_2) \indic{\frac{p(v_1, v_2)}{p(v_1)} \geq q_2 }\indic{p(v_1) \geq q_1} \\
    &= \sum_{v_1}\sum_{v_2} p(v_1, v_2) ( 1 - \indic{\frac{p(v_1, v_2)}{p(v_1)} < q_2 }) ( 1 - \indic{p(v_1) < q_1}) & \\
    &=1 + \sum_{v_1}\sum_{v_2} p(v_1, v_2) \left( \indic{\frac{p(v_1, v_2)}{p(v_1)} < q_2 }\indic{p(v_1) < q_1}  - \indic{\frac{p(v_1, v_2)}{p(v_1)} < q_2 } - \indic{p(v_1) < q_1}\right)&\\
    &\geq 1 - \sum_{v_1} p(v_1)\indic{p(v_1)< q_1} 
    - \sum_{v_1} \sum_{v_2} p(v_1, v_2) \indic{\frac{p(v_1, v_2)}{p(v_1)} < q_2 }&\\
    &=1 - \exposuresymbol_1(q_1)
    - \sum_{v_1} \sum_{v_2} p(v_1, v_2) \indic{\frac{p(v_1, v_2)}{p(v_1)} < q_2 },&
\end{flalign*}
where we used the fact that $\sum_{v_2}p(v_1, v_2)= p(v_1)$ for the second to last equality.

Combining this bound with \eqref{eq:initial_eq} and rearranging terms we have:
\begin{equation}
\label{eq:conditional_bound}
    \exposuresymbol_J(q_1 q_2) \leq \exposuresymbol_1(q_1) + \sum_{v_1} \sum_{v_2} p(v_1, v_2) \indic{\frac{p(v_1, v_2)}{p(v_1)} < q_2 }
\end{equation}
By Lemma~\ref{lemma:hq}, we have the following upper bound:
\begin{align*}
   \sum_{v_1} \sum_{v_2} p(v_1, v_2) \indic{\frac{p(v_1, v_2)}{p(v_1)} < q_2 } &=
   \sum_{v_2} \sum_{v_1} p(v_1, v_2) \indic{\frac{p(v_1, v_2)}{p(v_1)} < q_2 } \\
   &\leq \sum_{v_2}\sum_{v_1} p(v_1, v_2) G_{q_2}\left( \frac{p(v_1, v_2)}{p(v_1)}\right)
\end{align*}

We can now apply Lemma~\ref{lemma:logsum} and use the fact that $\sum_{v_1}p(v_1) = 1$ and $\sum_{v_1} p(v_1, v_2) = p(v_2)$ to upper bound the previous expression by:
\begin{equation*}
    \sum_{v_2} p(v_2)G_{q_2}(p(v_2)) = 
    \sum_{v_2} p(v_2)G_{q_2}(p(v_2)) \indic{p(v_2) < q_2} + 
    \sum_{v_2} p(v_2)G_{q_2}(p(v_2)) \indic{p(v_2) \geq q_2}
\end{equation*}

However, we know that $p(v_2)G_{q_2}(p(v_2))= p(v_2)$ for $p(v_2) < q_2$ and we can upper bound $p(v_2) G_{q_2}(p(v_2))$ by $q_2$. 
Therefore, we have: 
\begin{align*}
    \sum_{v_2} p(v_2)G_{q_2}(p(v_2))& \leq 
    \sum_{v_2} p(v_2)\indic{p(v_2)< q_2} + q_2 \indic{p(v_2) \geq q_2} \\
     &\leq \sum_{v_2} p(v_2)\indic{p(v_2)< q_2} + q_2 |V_2|= \exposuresymbol_2(q_2) + q_2 |V_2|
\end{align*}
Equation~\ref{eq:comp_two} can be obtained by replacing this bound in \eqref{eq:conditional_bound}.

To prove the general case, we use induction as follows.\\ 
Let $J$ denote an arbitrary index set of size $n \geq 2$, and let $j^*, j'$ denote two arbitrary elements of $J$.\\  
By equation \eqref{eq:comp_two}, we know that 
\begin{equation*}
    Q_J\left(\prod_{j \in J} t_j\right) \leq 
    Q_{J / \{j'\}}\left(\prod_{j \in J/\{j'\}} t_j)\right) + Q_{j'}(t_{j'}) +  t_{j'}|V_{j'}|.
\end{equation*}
Since $j^* \in J/\{j'\}$, 
we can apply induction to the first term on the right hand side of the above equation to obtain:
\begin{equation*}
    Q_J\left(\prod_{j \in J} t_j\right) \leq 
    \sum_{j \in J / \{j'\}} Q_j(t_j) + \sum_{j \in J/\{j'\}\colon j \neq j^*}t_j |V_j| +  Q_{j'}(t_j') +  t_{j'}|V_{j'}|.
\end{equation*}
The proof follows by rearranging terms in the above expression.

\end{proof}

\section{General composition rule (Theorem~\ref{Thm:CompositionTheoremC})}
\begin{proof}
Let \mbox{$f_j(i) = M_{ij}$} be the function that returns the value of column $j$ for row $i$.\\
Assume \mbox{$J = \{1, \ldots, k\}$}, and let \mbox{$f_1 \times \cdots \times f_k$} denote the function \mbox{$u \mapsto (f_1(u), \ldots, f_k(u))$}.\\
For any function $f$ with domain $U$, let \mbox{$V_f = \{f(u) : u \in U\}$} be the range of $f$.\\
Let \mbox{$g = f_1 \times \cdots \times f_k$}.

To prove the theorem, it suffices to show that if
\[
\left|\left\{u \in U : \probdist_i(f_i(u)) \ge \frac{1}{2^b_i}\right\}\right| \ge (1 - \delta_i)n.
\]
for all $i$, then
\[
\left|\left\{u \in U : \probdist_J(g(u)) \ge \frac{c}{2^{\sum_i b_i}}\right\}\right| \ge (1 - \sum_i \delta_i - c)n.
\]

Let \mbox{$V^+_i = \left\{v \in V_i : \probdist_i(v) \ge \frac{1}{2^{b_i}}\right\}$} and \mbox{$U^+_i = \{u \in U : f_i(u) \in V^+_i\}$}. \\
Clearly 
\begin{align*}
\mbox{$|V^+_i| \le 2^{b_i}$} \text{ and  } \mbox{$|U^+_i| \ge (1 - \delta_i)n$}.
\end{align*}

Let\mbox{ $V^+_g = (V^+_1 \times \cdots \times V^+_k) \cap V_g$} and \mbox{$U^+_g = \{u \in U : g(u) \in V^+_g\}$}. \\We immediately have
\begin{align*}
\mbox{$|V^+_g| \le 2^{\sum_i b_i}$} \text{ and  } \mbox{$|U^+_g| \ge \left(1 - \sum_i \delta_i\right)n$}
\end{align*}
the latter by taking a union bound.

Let \mbox{$V^-_g = \left\{v \in V^+_g : \probdist_J(v) < \frac{c}{2^{\sum_i b_i}}\right\}$} and \mbox{$U^-_g = \{u \in U^+_g : g(u) \in V^-_g\}$}.\\
Since \mbox{$|V^+_g| \le 2^{\sum_i b_i}$}, we have 
\begin{align*}
\mbox{$|U^-_g| = \left(\sum_{v \in V^-_g} \probdist_J(v)\right)n \le cn$}.
\end{align*}
Thus:
\[
\left|\left\{u \in U : p_J(g(u)) \ge \frac{c}{2^{\sum_i b_i}}\right\}\right| \ge |U^+_g| - |U^-_g| \ge \left(1 - \sum_i \delta_i - c\right)n
\]
completing this proof. 
\end{proof}

\section{Tightness of these composition theorems (Theorem~\ref{Thm:SlackNecessaryOnComposition})}
\begin{proof}
Let \mbox{$a = \frac{3}{c}$}, which is an integer by assumption.\\
Let $M$ be a matrix with $a2^k$ rows and $k$ columns, where $k$ is a positive integer whose value will be specified below.\\
And, let $J = \{1, \ldots, k\}$ and $J_i = \{i\}$ for each $i \in J$. In other words, $J$ contains all the columns of $M$.

For simplicity, assume that the values of the entries of $M$ belong to the set \mbox{$\{0, 1, \perp\}$} ($\perp$ for redacted). \\
Then, partition the $a2^k$ rows of $M$ into $2^k$ groups, with $a$ rows per group. All the rows in each group will contain identical values for every column, with $k$ exceptions. \\
Specifically, number the groups from $0$ to $2^k - 1$, and assign the binary encoding of $i$ to the $i^{\text{th}}$ group, with one bit per column. 
However, for $k$ arbitrarily chosen groups, replace one of the columns in one of the rows with $\perp$, choosing a different column for each row.

Observe that, in each column, $a2^{k - 1}$ rows are assigned one of the values in $\{0, 1\}$, and $a2^{k - 1} - 1$ rows are assigned the other value, with the remaining row assigned $\perp$. Thus
\[
\exposuresymbol_i\left(\frac{a2^{k-1} - 1}{a2^k}\right) = \frac{1}{a2^k}
\]
for all $i \in J$.

Also, since $a \ge 3$, we know that exactly $k$ rows are unique (specifically, the $k$ rows containing a $\perp$), and thus
\[
\exposuresymbol_J\left(\frac{2}{a2^k}\right) = \frac{k}{a2^k}.
\]

Now, choose $k$ large enough so that
\[
c\left(\frac{a2^{k - 1} - 1}{a2^k}\right)^k \ge \frac{2}{a2^k},
\]
which is possible since the limit of the ratio of the two sides of this inequality is less than $1$:
\[
\lim_{k \rightarrow \infty} \frac{\frac{2}{a2^k}}{c\left(\frac{a2^{k - 1} - 1}{a2^k}\right)^k} = \lim_{k \rightarrow \infty}  \frac{\frac{1}{a2^{k-1}}}{c\left(\frac12 - \frac{1}{a2^k}\right)^k} = \lim_{k \rightarrow \infty} \frac{1}{ac} \cdot \frac{1}{\left(\frac12 - \frac{1}{a2^k}\right)} \cdot \left(\frac{\frac12}{\frac12 - \frac{1}{a2^k}}\right)^{k-1} = \frac13 \cdot 2 \cdot 1 = \frac23.
\]

Let $q_i = \frac{a2^{k -1} - 1}{a2^k}$ for all $i \in J$.\\
Putting everything together, we have
\vspace{-3pt}
\begin{gather*}
\exposuresymbol_J\left(c\prod_i q_i\right) = \exposuresymbol_J\left(c\left(\frac{a2^{k - 1} - 1}{a2^k}\right)^k\right) \ge \exposuresymbol_J\left(\frac{2}{a2^k}\right) = \frac{k}{a2^k}\\ = \sum_i \frac{1}{a2^k} = \sum_i \exposuresymbol_i\left(\frac{a2^{k - 1} - 1}{a2^k}\right) = \sum_i \exposuresymbol_i(q_i)
\end{gather*}
\vspace{-3pt}
where the inequality follows because $\exposuresymbol_J(q)$ is monotonically non-decreasing in $q$.
\end{proof}

\section{Derivation of the statistical exposure}
\label{Appendix:statistical_exposure}

First, recall that the cumulative distribution of a Binomial distribution with parameters $n$ and $p$ is:
\begin{align*}
    F(k; n, p) = P(X \leq k) = \sum_{i=0}^{k} \binom{n}{i} p^{i}(1-p)^{n-i} = I_{1-p}\big(n-k,k+1 \big). \label{eq:CDFBinomial}
\end{align*}

The statistical exposure $\smoothexposuresymbol_{\probdist}(n,k)$  is defined as the probability that a random user in a database composed of $n$ users sampled i.i.d. from $\probdist$ is less than \mbox{$k$-anonymous}.  
That is: 
\vspace{-4pt}
\begin{align*}
      \smoothexposuresymbol_{\probdist}(n,k) &= \frac{1}{n} \sum_{i=1}^{|V|}  \sum_{j=0}^{k-1} \binom{n}{j}  j   p_i^{j} (1-p_i)^{n-j} \\
      &= \frac{1}{n} \sum_{i=1}^{|V|}  \sum_{j=1}^{k-1} j\frac{n!}{j!(n-j)!} p_i^{j} (1-p_i)^{n-j} \\
      &= \frac{1}{n} \sum_{i=1}^{|V|}  \sum_{j=1}^{k-1} n \frac{(n-1)!}{(j-1)!(n-j)!}  p_i^{j} (1-p_i)^{n-j}  \\
      &=  \sum_{i=1}^{|V|}  \sum_{j=1}^{k-1}  \binom{n-1}{j-1}  p_i^{j} (1-p_i)^{n-j}  
\end{align*}
where on the second line we start the sum over $j$ at $1$ as the term associated with $j=0$ is zero. 

Let \mbox{$j' = j-1$} and \mbox{$n' = n-1$}:
\vspace{-4pt}
\begin{align*}
    \smoothexposuresymbol_{\probdist}(n,k) &= \sum_{i=1}^{|V|}  \sum_{j'=0}^{k-2}  \binom{n'}{j'}  p_i^{j'+1} (1-p_i)^{(n'-1)-(j'-1)}  \\
      &= \sum_{i=1}^{|V|}  \sum_{j'=0}^{k-2}  p_i \binom{n'}{j'}  p_i^{j'} (1-p_i)^{(n'-j')}  \\
      & = \sum_{i=1}^{|V|}   p_i I_{1-p_i}\big( n'-(k-2), (k-2) + 1\big) \\
      & = \sum_{i=1}^{|V|}   p_i I_{1-p_i}\big( (n-1)-(k-2), k-1\big) \\
      & = \sum_{i=1}^{|V|}   p_i I_{1-p_i}\big( n-(k-1), k-1\big).
\end{align*}

\section{Error bound for the \mbox{plug-in} estimator of the exposure (Theorem~\ref{thm:plugin-exposure})}

\begin{proof}
Note that if $|\probdist_j(v) - \widehat{\probdist}_j(v)| \le \gamma$ for all $v \in V_j$ then
\begin{align*}
\exposuresymbol_j(t) &= \sum_{v \in V_j} \probdist_j(v) \cdot \indic{\probdist_j(v) < t}\\
&\le \sum_{v \in V_j} (\widehat{\probdist}_j(v) + \gamma) \cdot \indic{\widehat{\probdist}_j(v) - \gamma < t}\\
&= \sum_{v \in V_j} \widehat{\probdist}_j(v) \cdot \indic{\widehat{\probdist}_j(v) < t + \gamma} +  \gamma  \sum_{v \in V_j} \indic{\widehat{\probdist}_j(v) < t + \gamma}\\
&\le \widehat{\exposuresymbol}_j(t + \gamma) + \gamma |V_j|
\end{align*}
and
\begin{align*}
\exposuresymbol_j(t) &= \sum_{v \in V_j} \probdist_j(v) \cdot \indic{\probdist_j(v) < t}\\
&\ge \sum_{v \in V_j} (\widehat{\probdist}_j(v) - \gamma) \cdot \indic{\widehat{\probdist}_j(v) + \gamma < t}\\
&= \sum_{v \in V_j} \widehat{\probdist}_j(v) \cdot \indic{\widehat{\probdist}_j(v) < t - \gamma} - \gamma \sum_{v \in V_j} \indic{\widehat{\probdist}_j(v) < t - \gamma}\\
&\ge \widehat{\exposuresymbol}_j(t - \gamma) - \gamma |V_j|.
\end{align*}
Moreover, by Hoeffding's inequality and the union bound, we have
\[
\Pr\left[\max_{v \in V_j} |\probdist_j(v) - \widehat{\probdist}_j(v)| > \gamma\right] \le |V_j|\exp(-2m\gamma^2). \qedhere
\]
\end{proof}

\section{Hardness of estimating the exposure using any estimator (Theorem~\ref{thm:lower})}

We use the following form of Le Cam's theorem \cite{le1960approximation}:
\begin{theorem}\label{thm:lecam}
Suppose there exist $\probdist_0$ and $\probdist_1$ from some parametric family of distributions $\mathcal{P}$ such that $\text{KL} ( \probdist_0\vert \vert \probdist_1 ) \le \tfrac{\log 2}{n}$. 
Then
\[
R_n = \inf_{\widehat{\theta} } \sup_{\probdist\in \mathcal{P}} \mathbb{E} [d(\widehat{\theta}(X_1, \dots, X_n ), \theta ( \probdist ) ) ] \ge \frac{d( \theta ( \probdist_0 ), \theta ( \probdist_1) )}{16}
\]
where $d$ is a semi-metric.
\end{theorem}

\begin{proof}
Let us fix $d$ as $d(x,y) = |x-y|$ and let $\mathcal{P} = \Delta_{s+1}$ be the set of discrete distributions represented by the $s+1$ dimensional probability simplex.\\ 
Next, we define $\probdist_0 = (p_{0,1}, \dots,  p_{0,s+1} )$ and $\probdist_1 = (p_{1,1}, \dots,  p_{1,s+1} )$ as
\[
p_{0,i} =
\begin{cases}
\frac{1}{s^2}-\varepsilon & \text{~if~} i \in [s] \\
\frac{s-1+s^2\varepsilon}{s} & \text{~if~} i = s+1
\end{cases}
\]
and
\[
p_{1,i} =
\begin{cases}
\frac{1}{s^2} & \text{~if~} i \in [s] \\
\frac{s-1}{s} & \text{~if~} i = s+1
\end{cases}
\]
where $\varepsilon \le  \tfrac{1}{s^2}$. \\
We use 
\[
q = \frac{1}{s^2}-\frac{\varepsilon}{2} \enspace 
\]
with this at hand, we can show that the L1 difference of exposure for $\probdist_0$ and $\probdist_1$ is 
\[
d( \exposure{\probdist_0}{q}, \exposure{\probdist_1}{q}) = \left\vert \frac{1}{s} - s\varepsilon \right\vert
\]
since $p_{0,s+1}>q$ as
\begin{align*}
    p_{0,s+1} - q 
    & = \frac{s-1+s^2\varepsilon}{s} - \frac{1}{s^2}+\frac{\varepsilon}{2} \\
    & \ge \frac{s^2 -s - 1}{s^2}  \\
    & = 1- \frac{s+1}{s^2}
\end{align*}
which is positive when $s\ge2$ and $\exposure{\probdist_1}{q} = 0$ as $\tfrac{s-1}{s} \ge q = \tfrac{1}{s^2} - \tfrac{\varepsilon}{2}$. 

Next, we upper bound the KL divergence as
\begin{align}
\text{KL} (\probdist_0\vert \vert\probdist_1 ) 
    &= s \left(\frac{1}{s^2} - \varepsilon \right) \log \frac{1/s^2 - \varepsilon}{1/s^2} + \frac{s-1+s^2\varepsilon}{s} \log \left(\frac{s-1+s^2\varepsilon}{s} \cdot \frac{s}{s-1}\right) \notag \\
    & = \left( \frac{1}{s} - s\varepsilon \right) \log (1-s^2\varepsilon ) + \left(1+s\varepsilon - \frac{1}{s} \right) \log \left( 1 + \frac{s^2\varepsilon}{s-1}\right) \notag \\
    & = \left( \frac{1}{s} - s\varepsilon \right) \left( \log (1-s^2\varepsilon ) - \log \left( 1 + \frac{s^2\varepsilon}{s-1}\right) \right) + \log \left( 1 + \frac{s^2\varepsilon}{s-1}\right) \notag \\
    & = \left( \frac{1}{s} - s\varepsilon \right) \log \left( 1-\frac{s^3\varepsilon}{s-1+s^2\varepsilon} \right) + \log \left( 1 + \frac{s^2\varepsilon}{s-1}\right) \notag \\
    & \le \left( \frac{1}{s} - s\varepsilon \right)\frac{s^3\varepsilon}{1-s-s^2\varepsilon} + \frac{s^2\varepsilon}{s-1} \label{eq:aux_klupper} \\
    & \le \frac{s^2\varepsilon}{1-s-s^2\varepsilon} + \frac{s^2\varepsilon}{s-1} \notag \\
    & =  s^2 \varepsilon \left( \frac{1}{s} + \frac{1}{s-1} \right)\notag \\
    & \le  \frac{2s^2 \varepsilon}{s-1} \notag
\end{align}
where (\ref{eq:aux_klupper}) follows from the fact that $\log(1+x) \le x$ for $x> -1$. 

Setting \mbox{$\varepsilon = \frac{s-1}{s^2 n}$} allows us to apply Theorem \ref{thm:lecam}, and we are done.
\end{proof}

\section{Error bound for the \mbox{plug-in} estimator of the statistical exposure (Theorem~\ref{thm:approximation_smooth})}
\label{prooftheorem6}
\begin{proof}
We start by noticing that the regularized incomplete betafunction \mbox{$I_{1 - p_i}(n -(k-1), k-1) := F(k; n, p)_i)$} corresponds to the  cumulative distribution of a binomial random variable $\text{Bin}(n, p_i)$ with parameters $n$ and $p_i$.\\
Thus, by definition of $\smoothexposuresymbol$ we have:
\begin{align*}
    \smoothexposuresymbol_{\probdist}(k, n) - \smoothexposuresymbol_{\widehat{\probdist}}(k,n) &= 
    \sum_{i=1}^{|V|} p_i F(k; n, p)_i) - \widehat{p_i} F(n, k, \widehat{p}_i) \\
    &=\sum_{i=1}^{|V|} p_i \left(F(k; n, p)_i) - F(n, k, \widehat{p}_i) \right)
    + F(n, k, \widehat{p}_i) (p_i - \widehat{p}_i)\\
\end{align*}
\vspace{-4pt}
So, using the fact that \mbox{$F(k; n, p_i)) \leq 1$}, we have that 
\begin{align}
    |\smoothexposuresymbol_{\probdist}(k, n) - \smoothexposuresymbol_{\widehat{\probdist}}(k,n)|
    & \leq 
    \sum_{i=1}^{|V|} p_i \left| F(k; n, p)_i) - F(n, k, \widehat{p}_i) \right|
     + \|\probdist - \widehat{\probdist}\|_1 \nonumber\\
    & \leq 
    \sum_{i=1}^{|V|} p_i \left| F(k; n, p)_i) - F(n, k, \widehat{p}_i) \right|
     + |V|\|\probdist - \widehat{\probdist}\|_\infty \label{eq:qpbound}
\end{align}
We now proceed to bound the first term in the above equation, by Section 2.2 of \cite{AdJo06}, we have:
\begin{align*}
    |F(k; n, p)_i) - F(n, k, \widehat{p}_i| &\leq d_{\text{TV}}(\text{Bin}(n, p_i), \text{Bin}(n, \widehat{p}_i)) \\
    &\leq \frac{\sqrt{e}}{2} \frac{\tau(|p_i - \widehat{p}_i|)}{1 - \tau(|p_i - \widehat{p}_i|)^2} \\
   & \leq \frac{\sqrt{e}}{2} \tau(|p_i - \widehat{p}_i|).
\end{align*}
where $d_{\text{TV}}$ denotes the total variation distance and \mbox{$\tau(x) := x \sqrt{\frac{n+1}{p_i ( 1- p_i)}}$}.

Replacing this bound in \eqref{eq:qpbound}, we obtain:
\begin{align*}
    |\smoothexposuresymbol_{\probdist}(k, n) - \smoothexposuresymbol_{\widehat{\probdist}}(k,n)|
    &\leq \frac{\sqrt{e}}{2}\sum_{i=1}^{|V|} \sqrt{\frac{p_i (n+1)}{1 - p_i}}|p_i - \widehat{p_i}| + |V|\|\probdist - \probdist\|_\infty \\
    &\leq \frac{\sqrt{e}}{2} \|\probdist - \widehat{\probdist}\|_\infty \left( \sum_{i=1}^{|V|} \sqrt{\frac{p_i (n+1)}{1 - p_i}} + |V| \right) \\
    &\leq  |V|\left( \frac{\sqrt{e (n + 1)}}{2\sqrt{|V| - 1}} + 1 \right) \|\probdist - \widehat{\probdist}\|_\infty
    \end{align*}
\end{proof}

\section{Relationships between exposure and entropy}
\label{appendix:connectionentropy}
\subsection{Proof of Proposition~\ref{prop:entropyconnection}}
\begin{proof}
By definition of exposure $Q(t)$ we have:
\begin{align*}
  \int_0^1 \frac{Q(t)}{t}dt &= \int_0^1 \frac{1}{t} \sum_{i=1}^n p_i \indic{p_i < t}dt \\ 
  &=\sum_{i=1}^n p_i \int_0^1 \frac{\indic{p_i < t}}{t} dt\\
  &=\sum_{i=1}^n p_i \int_{p_i}^1 \frac{1}{t}dt = -\sum_{i=1}^n p_i \log p_i = H(\probdist)
  \end{align*}
This proves the first statement of the proposition. 
To prove the upper bound on the exposure we use the fact that
\begin{align*}
    Q(t) &= P_{I \sim \probdist}(p_I < t) 
    =P_{I \sim \probdist}(-\log(p_I) > -\log t) \\
    & \leq\frac{ E_{I \sim \probdist}[-\log p_I]}{- \log t} 
    = \frac{H(\probdist)}{- \log t},
\end{align*}
where we have used Marokv's inequality. 
\end{proof}
We now show that the above bound is tight.

\begin{proposition}
Let \mbox{$B> 0$, $ 1>  t> 0$} be such that \mbox{$-\frac{B}{\log t} < 1 - t$} and \mbox{$-\frac{B}{t \log t}$} is an integer greater or equal to $1$.\\
There exists a distribution $\probdist$ such that  
\begin{align*}
H(\probdist) \leq  B + \frac{1}{e} \quad  \text{ and } \quad Q(t) = \frac{B}{- \log t}.  
\end{align*}
\end{proposition}

\begin{proof}
Let \mbox{$n = \frac{B}{-t \log t}$}, and let \mbox{$\probdist \in \mathbb{R}^{n+1}$} be defined as:\\
\begin{align*}
p_i &= t  \quad \text{if } i \leq n \quad \text{ and} \quad 
p_{n+1} = 1 - nt= 1 + \frac{B}{\log t}. 
\end{align*}
Note that since \mbox{$1  + \frac{B}{\log t}  > t$}, it follows that \mbox{$Q(t) = \sum_{i=1}^n p_i = nt = -\frac{B}{\log t}$}. 

On the other hand, the entropy of this distribution is given by: 
\begin{equation*}
    -\sum_{i=1}^n p_i \log p_i - p_{n+1}\log p_{n+1}
     = n t \log t - p_{n+1} \log p_{n+1}
     = B - p_{n+1} \log p_{n+1}
\end{equation*}
The result of the proposition follows from the fact that the function \mbox{$x \mapsto -x \log x$} achieves a maximum value of \mbox{$\frac{1}{e}$} in $[0,1]$. 
\end{proof}

\end{document}